\documentclass[conference]{IEEEtran}
\IEEEoverridecommandlockouts
\usepackage{amsthm}
\usepackage{mathrsfs}
\usepackage{algorithm}
\theoremstyle{definition}
\newtheorem{definition}{Definition}[section]
\newtheorem{theorem}{Theorem}[section]

\theoremstyle{remark}

\usepackage{cite}
\usepackage{amsmath,amssymb,amsfonts}
\usepackage{algorithmic}
\usepackage{graphicx}
\usepackage{tikz}
\usetikzlibrary{arrows.meta, positioning}
\usepackage{textcomp}
\usepackage{xcolor}
\def\BibTeX{{\rm B\kern-.05em{\sc i\kern-.025em b}\kern-.08em
    T\kern-.1667em\lower.7ex\hbox{E}\kern-.125emX}}
\begin{document}

\title{SoK: Speedy Secure Finality
}

\author{\IEEEauthorblockN{Yash Saraswat}
\textit{Indian Institute of Technology, Roorkee, India}
\and
\IEEEauthorblockN{Abhimanyu Nag}
\textit{University of Alberta, Canada}}
% \and
% \and
% \and
% \and

\maketitle

\begin{abstract}
While Ethereum has successfully achieved dynamic availability along with safety, a fundamental delay remains between block production and immutable finality. In Ethereum's current Gasper protocol, this latency extends to approximately 15 minutes, exposing the network to ex-ante reorg attacks, leading to MEV capture and limiting the efficiency of economic settlement. This inherent delay along with cracks in Gasper has catalyzed the emergence of ``Speedy Secure Finality" (SSF): research efforts dedicated to minimizing confirmation times without compromising rigorous security guarantees. This work presents a survey of the state of the art on the evolution of fast finality protocols. We define the core theoretical primitives necessary for understanding this landscape, including reorg resilience and the generalized sleepy model, and trace their development from Goldfish to RLMD-GHOST. We then analyze the practical trade-offs of Single Slot Finality against the communication bottlenecks of large validator sets. Finally, we survey the 3-Slot Finality (3SF) protocol as a pragmatic convergence of these efforts, bridging the gap between theoretical safety guarantees and the engineering constraints of the Ethereum network.
\end{abstract}


\section{Introduction}

Consensus protocols enable a distributed system vis-à-vis blockchains to maintain a consistent replicated state in the presence of faults, adversarial behavior, and network uncertainty. In blockchain systems, consensus further serves as the mechanism by which a globally ordered ledger is constructed under open participation and economic incentives. Two central properties governing the security and usability of such systems are \emph{availability}: the guarantee that block production does not halt as long an honest majority of validators are online, and \emph{finality}: the guarantee that once a block is confirmed, it will not be reverted in any admissible execution. A detailed discussion about consensus protocols can be found in \cite{bano2019sok,arafatsok}.
\\\\
Over the past two decades, consensus protocols have evolved along multiple, often conflicting design axes. \textit{Classical Byzantine Fault Tolerant} (BFT) (see \cite{zhong2023byzantine} for a comprehensive survey) protocols achieve deterministic finality under partial synchrony by assuming a fixed replica set with a bounded fraction of Byzantine faults. Nakamoto style longest chain protocols, by contrast, operate under dynamic participation and unknown resource levels, achieving adaptivity and liveness at the cost of replacing finality with probabilistic confirmation. This tension is fundamental where Lewis-Pye and Roughgarden \cite{resourcepool} formalize an impossibility result showing that adaptivity and finality cannot be simultaneously achieved in the unsized, partially synchronous setting.
\\\\
Modern proof of stake blockchains \cite{saleh2021blockchain} combine elements from both paradigms that separate availability from finality. Ethereum’s consensus protocol, Gasper \cite{gasper}, is the most prominent example, pairing a dynamically available fork-choice rule with a BFT-style finality gadget in an ebb-and-flow
architecture. While such designs have proven robust in practice, they introduce choices between confirmation latency, reorganization risk, and implementation complexity. Understanding these trade-offs, and the extent to which they can be optimized is the focus of this paper.
\\\\
\emph{Fast finality} (see \cite{anceaume2020finality} and applications in \cite{chou2025minimmit} and \cite{pan2021plume}) refers to protocols that aim to
minimize the latency between block proposal and irrevocable confirmation, subject to the constraints imposed by participation dynamics and network assumptions. Ethereum’s Gasper protocol achieves finality on the order of tens of slots, creating a window during which blocks are vulnerable to reorganization and economically motivated attacks. Recent work seeks to shrink this window by refining fork-choice rules, vote handling, and participation models, without abandoning dynamic availability. As far as the authors know, there is no concrete documentation of the fast finality design model in blockchains and this SoK systematically surveys this space through the taxonomy outlined above.
\\\\
We begin by formalizing the relevant network and fault
models, including partial synchrony and the sleepy model. We then analyze Ethereum’s Gasper protocol as a canonical ebb-and-flow construction. Building on this foundation, we examine recent proposals for faster finality, including Goldfish, RLMD-GHOST, single-slot finality and then three-slot finality protocols. For each class of designs, we identify the underlying assumptions, formal guarantees, and implementation constraints that govern their viability in large-scale permissionless systems.
\\\\
Our objective is not to propose a new protocol, but to clarify the structure of the design space, identify the fundamental trade-offs that shape it, and provide a principled framework for evaluating future proposals for speedy secure finality.

\section{Background}

This section outlines the established distributed systems primitives necessary to analyze consensus and finality in decentralized and permissionless networks. We aim to provide a succinct overview of the seminal research that defines the technical framework for achieving speedy secure finality.

\subsection{Network Assumptions}
When working with distributed systems and consensus protocols, there must be an assumption made about the network delay of messages. If a message is sent from one node to another, then it is important for us to make an assumption about when to expect a reply. If the connection between the nodes is damaged for some reason, it is possible we may never receive a reply. 
\\\\
Conversely, if a node is adversarial, it may wish to withhold that message for some definite or indefinite amount of time for malicious purposes. Both of these scenarios are indistinguishable and hence, it is necessary to assume some message delay to formalize any consensus protocol. Based on these assumption, a network can be classified into three models \cite{monge2003theories}:

\begin{definition}[Synchronous Network]
     A \textit{synchronous network} is one where the network delay $\Delta$ is finite and known i.e for any message, the adversary can delay a message by at most a known $\Delta < \infty$. 
\end{definition}
This type of model is easy to reason about in theory but is often not practical because of its strict assumptions. A longer time bound might compromise the latency of the network while a shorter one might make it too abstract and impractical for any real world application.

\begin{definition}[Asynchronous Network]
    An \textit{asynchronous network} is one where $\Delta$ is finite but unknown. This means that for any message, the adversary will deliver the message eventually but there is no assumption made about the time bound of the delay.
\end{definition}
As a result of this assumption, these protocols are often very robust but are more complex and harder to reason about. Such models are subject to fundamental constraints, most notably the \textit{Fischer-Lynch-Patterson (FLP) impossibility result} \cite{flpimposibility}.

\begin{definition}[Partially-Synchronous Network]
    A \textit{partially-synchronous network} is one where $\Delta$ is a finite and known time bound following a period known as the \textit{global stabilization time} $\text{GST}$ \cite{gst}
\end{definition}

A message sent at any given time $T$ must be delivered at most by time  $\Delta$ + $\text{max}(T, GST)$. In essence, the network behaves asynchronously until $\text{GST}$, after which it behaves synchronously. Crucially, the adversary can delay $\text{GST}$ by any finite but unknown time bound and there is no way for the network to know if $\text{GST}$ has occurred or not. This model aims to reconcile the safety of asynchronous systems with the liveness of synchronous ones by accounting for transient network instability.

\subsection{Practical Byzantine Fault Tolerance (pBFT)}

Practical Byzantine Fault Tolerance (pBFT) was introduced by Castro and Liskov \cite{pbft} in 1999 to achieve fault-tolerant state machine replication in an adversarial environment. It replicates a deterministic state machine across $n$ replicas of which at most $f$ may be Byzantine. The network is partially synchronous but supports authenticated reliable channels. The protocol’s fault tolerance limit is derived from the necessity to ensure a majority quorum of honest nodes under the worst-case scenario. Assuming $f$ nodes exhibit Byzantine behavior, up to $f$ distinct nodes may incur non-adversarial crash faults. The following summarizes the core mechanics of the protocol as defined in \cite{pbft}.

\begin{theorem}[Fault bound]
\label{thm:pbft-bound}
pBFT satisfies safety and liveness if and only if $n \ge 3f+1$.
\end{theorem}

\begin{proof}[Proof sketch]
The protocol operates by forming quorums of size $2f+1$. If $Q_1$ and $Q_2$ are such quorums, then $|Q_1 \cap Q_2| \ge f+1$ whenever $n \ge 3f+1$. Since at most $f$ replicas are faulty, every quorum intersection contains at least one honest replica. The bound $n \ge 3f+1$ guarantees that for any two quorums $Q_1,Q_2 \subseteq [n]$ with $|Q_1|=|Q_2|=2f+1$, we have $|Q_1 \cap Q_2| \ge f+1$, ensuring the presence of at least one honest replica in every quorum intersection. This property is necessary for agreement and becomes impossible when $n \le 3f$; hence the bound is tight.
\end{proof}
Each request $r$ is identified by a triple $(v,s,d)$ consisting of the view number $v$, a sequence number $s$, and digest $d=\mathrm{Digest}(r)$. For each $(v,s)$ the protocol ensures that all correct replicas that execute at position $s$ select the same digest. The pBFT algorithm is executed in 3 phases: $pre-prepare$, $prepare$ and $commit$.

\subsubsection{Pre-prepare} 
 The algorithm has a leader called the primary, $p$, responsible for starting the protocol. In this phase, the primary assigns a sequence number, $s$, to the request it received from the client, and broadcasts it to all the nodes in the network. This phase is for the nodes to acknowledge that a request with a sequence number exists and to validate the state replication. The primary of view $v$, denoted $p(v)$, broadcasts $(\mathsf{PRE\text{-}PREPARE}, v, s, d)$. An honest replica accepts $(v,s,d)$ only if it is in view $v$, the sequence number is valid, and it has not accepted a conflicting digest for $(v,s)$.

\subsubsection{Prepare}  
Upon accepting the pre-prepare request, the nodes now broadcast their validation to the network. A node is set to enter the prepared stage if it observes $2f + 1$ prepare messages. At this stage, the nodes know that a quorum of the network has validated the request. Therefore, upon accepting the pre-prepare request, replica $i$ broadcasts $\langle \mathsf{PREPARE}, v, s, d, i \rangle$.

\begin{definition}[Prepared]
\label{def:prepared}
A replica is prepared for $(v,s,d)$ once it has accepted the corresponding pre-prepare and has received prepare messages supporting $(v,s,d)$ from at least $2f$ distinct replicas. Equivalently, it has evidence of a quorum of size $2f+1$ supporting $(v,s,d)$.
\end{definition}

\subsubsection{Commit}  
After nodes enter the prepared state, they broadcast the commit message to the network. A node is set to enter the commit stage if it observes more than $2f + 1$ commit messages. At this stage, the node knows that a quorum of the network has committed the request and it is said to be committed. A prepared replica broadcasts $\langle \mathsf{COMMIT}, v, s, d, i \rangle$.

\begin{definition}[Committed-Local]
\label{def:committed-local}
A replica is committed-local for $(v,s,d)$ once it is prepared for $(v,s,d)$ and has received at least $2f+1$ commit messages supporting $(v,s,d)$ from distinct replicas.
\end{definition}

A replica that is committed-local for $(v,s,d)$ executes $r$ at sequence number $s$ and applies it to its state machine.

\begin{figure}[h]
    \centering
    \includegraphics[width=1\linewidth]{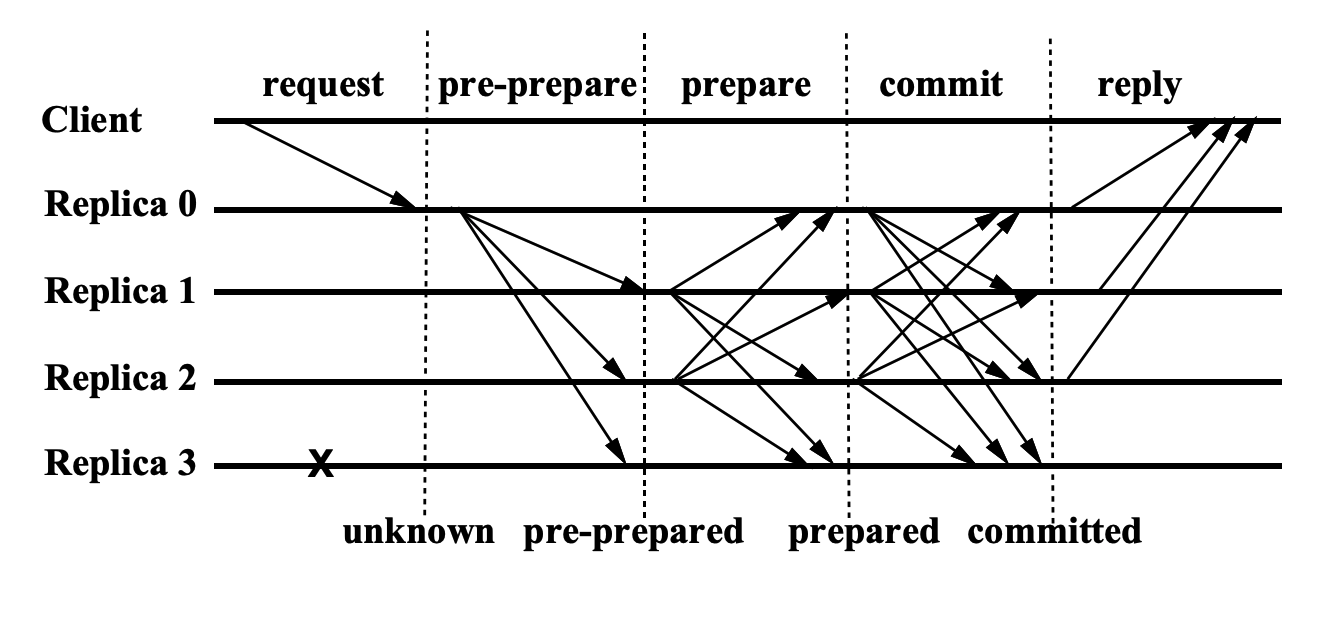}
    \caption{Normal case execution with Replica 0 as leader \cite{pbft}}
    \label{fig:pbft_diagram}
\end{figure}

\begin{theorem}[Safety]
\label{thm:pbft-safety}
No two correct replicas can commit different digests at the same sequence number.
\end{theorem}

\begin{proof}[Proof sketch]
If replicas commit $(v,s,d)$ and $(v',s,d')$ with $d \neq d'$, each commit is supported by a quorum of size $2f+1$. Their intersection contains at least one correct replica. A correct replica cannot emit commit messages for two different digests at the same $(v,s)$, hence, $d=d'$.
\end{proof}

\begin{theorem}[Liveness]
\label{thm:pbft-liveness}
Under partial synchrony, pBFT eventually selects a correct primary through the view-change mechanism, after which every valid client request is eventually pre-prepared, prepared, committed, and executed.
\end{theorem}

Together, safety and liveness establish pBFT as a complete Byzantine fault tolerant replication protocol under the $n \ge 3f+1$ bound and partial synchrony. Its reliance on a fixed, fully active replica set stands in contrast to consensus models with dynamic or intermittently online participants. We next examine the Sleepy Model of Consensus, which relaxes these assumptions and enables agreement in settings with fluctuating node availability.

\subsection{Sleepy Model of Consensus}

In permissionless blockchain settings, replicas need not remain continuously active. Pass and Shi \cite{sleepy} formalize this by introducing the notion of \emph{sleepy} behavior, in contrast to the fully synchronous participation assumed by classical BFT protocols such as pBFT.

\begin{definition}[Sleepy Node {\cite{sleepy}}]
A replica $i$ is \emph{sleepy} if at each time step it may be either \emph{online} (awake) or \emph{offline} (asleep), where the online/offline state may change arbitrarily throughout the execution. An offline replica produces no messages and, from the perspective of the protocol, is indistinguishable from an adversarial replica.
\end{definition}

Protocols based on static quorum assumptions such as pBFT, which requires $n \ge 3f+1$ and therefore an online honest quorum of size $\frac{2}{3}n$, lose liveness whenever the set of online honest replicas falls below the quorum threshold. Thus classical BFT consensus is not robust under dynamic participation.
\\\\
Pass and Shi introduce a consensus model in which the relevant security parameter is the fraction of \textit{honest online} replicas at each time step. Their protocol achieves safety provided that adversarial online participation never exceeds the honest-online majority and achieves liveness under partial network synchrony given the same condition. This relaxes the static quorum requirement and makes consensus compatible with fluctuating availability.
The Ouroboros proof-of-stake protocol \cite{ouroboros} later termed this robustness property \textit{dynamic availability}.

\begin{definition}[Dynamic Availability]
A consensus protocol is \emph{dynamically available} if it maintains liveness (i.e., honest transactions are eventually confirmed) whenever, at every time step, the set of online honest replicas has a majority over online adversarial replicas.
\end{definition}

Dynamic availability does not relax safety: consistency holds only while the honest-online majority assumption remains valid. Once adversarial online participation dominates, both Sleepy Consensus and Ouroboros may experience safety violations due to adversarial control of leader election or chain extension.

\subsection{Finality and Adaptivity}

Since protocols like pBFT do not assume dynamic participation, there exists zero chance that a confirmed block can be rolled back even if a majority of nodes crash and the network is halted. This property is termed as finality. Protocols like Nakamoto consensus do not have this property.
\\\\
Nakamoto demonstrated in the Bitcoin whitepaper\cite{bitcoin} that the probability that an attacker can re-org a block decreases exponentially as new blocks are confirmed on the canonical chain. This sort of guarantee is called probabilistic finality. We follow the framework of Lewis-Pye and Roughgarden \cite{resourcepool} to formalize the notions of finality and adaptivity for permissionless consensus protocols.

\paragraph{Confirmation, liveness, and security.}
Let $P = (I,O,C)$ denote an \emph{extended protocol}, where $I$ is the instruction set executed by honest nodes, $O$ is the permitter oracle governing when a node is allowed to publish state updates (blocks, votes, etc.), and $C$ is a \emph{confirmation rule}. For any message state $M$ (the set of all messages delivered to a node), $C(M)$ returns a chain of \emph{confirmed} blocks, depending on a security parameter $\varepsilon>0$.

For an honest node $U$ and timeslots $t_1<t_2$, let $M_i$ be $U$'s message state at time $t_i$, and write $C(M_i)$ for its confirmed chain. The interval $[t_1,t_2]$ is a \emph{growth interval} for $U$ if
\[
  |C(M_2)| \;>\; |C(M_1)|,
\]
that is, if the number of confirmed blocks strictly increases.

\begin{definition}[Liveness {\cite[Def.~3.1]{resourcepool}}]
A protocol is \emph{live} if, for every $\varepsilon>0$, there exists $\ell_\varepsilon \in \mathbb{N}$ such that the following holds with probability at least $1-\varepsilon$ in the underlying probabilistic model: for every honest user $U$ and every pair of timeslots $t_1<t_2$ with $t_2-t_1\ge \ell_\varepsilon$, if the interval $[t_1,t_2]$ is entirely synchronous, then $[t_1,t_2]$ is a growth interval for $U$.
\end{definition}

Intuitively, liveness requires that, during sufficiently long synchronous periods, the set of confirmed blocks continues to grow for every honest node. Two blocks are \emph{compatible} if one is an ancestor of the other; otherwise they are \emph{incompatible}. If a block $B$ belongs to $C(M)$ for the message state $M$ of user $U$ at time $t$, we say that $B$ is \emph{confirmed for $U$ at time $t$}.

\begin{definition}[Security {\cite[Def.~3.3]{resourcepool}}]
A protocol is \emph{secure} if for every $\varepsilon>0$, for all honest users $U_1,U_2$ and all times $t_1,t_2$, with probability at least $1-\varepsilon$ every block confirmed for $U_1$ at time $t_1$ is compatible with every block confirmed for $U_2$ at time $t_2$.
\end{definition}

Security formalizes the idea that confirmation defines a single coherent ledger prefix, up to negligible failure probability.

\paragraph{Finality.}
The framework distinguishes between synchronous and \emph{partially synchronous} network models. In the partially synchronous setting, the duration is partitioned into synchronous and asynchronous intervals, with the adversary controlling message delivery during asynchronous intervals. A protocol has \emph{finality} precisely when its confirmation rule remains secure even in the presence of unbounded network failures after confirmation.

\begin{definition}[Finality {\cite[Def.~3.4]{resourcepool}}]
A protocol \emph{has finality} if it is secure in the partially synchronous setting.
\end{definition}

Classical BFT-style protocols such as pBFT and Algorand are designed so that, once a block is confirmed (e.g., by collecting a quorum of committee signatures), the probability that it is ever revoked is bounded by the security parameter, even if honest participation ceases or the network subsequently partitions for an arbitrarily long time. In this formal sense, they achieve finality.
\\\\
By contrast, longest-chain Nakamoto-style protocols (e.g., Bitcoin) do not satisfy finality: confirmation is only secure under ongoing participation of an honest majority in a sufficiently synchronous regime. As shown in the Bitcoin whitepaper~\cite{bitcoin}, for an attacker with a resource fraction $q<\tfrac12$, the probability that a $k$-deep block can be rolled back decays exponentially in $k$, yielding \emph{probabilistic finality} but not finality in the above sense.

\paragraph{Resource pools, adaptivity, and the sized/unsized distinction.}
Protocols in~\cite{resourcepool} are analyzed relative to a \emph{resource pool} $R$, which assigns to each public key $U$ and timeslot $t$ a nonnegative resource balance $R(U,t,M)$ (e.g., hashrate in PoW or stake in PoS), possibly depending on the message state $M$. Let
\[
  T(t,M) \;=\; \sum_U R(U,t,M)
\]
denote the total resource at time $t$ in state $M$.

Two resource settings are distinguished:
\begin{itemize}
  \item \emph{Sized setting:} the total resource $T$ is a predetermined function of $(t,M)$ available to the protocol and permitter.
  \item \emph{Unsized setting:} $R$ (and hence $T$) is undetermined, subject only to coarse bounds $T(t,M)\in[I_0,I_1]$ and adversarial resource bounds.
\end{itemize}

Informally, PoS protocols with on-chain stake (e.g., Algorand) are naturally modeled in the sized setting, while PoW protocols (e.g., Bitcoin) live in the unsized setting, where the total hashrate can fluctuate unpredictably.

\begin{definition}[Adaptivity {\cite[Def.~3.2]{resourcepool}}]
A protocol is \emph{adaptive} if it is live in the unsized setting.
\end{definition}

Thus, a protocol is adaptive if it remains live without knowing the total resource $T(t,M)$. It automatically adjusts to arbitrary (but bounded) changes in participation. Longest-chain protocols like Bitcoin and Snow White are adaptive in this sense: as long as some honest resource remains, blocks continue to be produced and confirmed at some rate without explicit reconfiguration.

\paragraph{Impossibility of combining finality and adaptivity.}
Lewis-Pye and Roughgarden prove a CAP-style impossibility theorem:

\begin{theorem}[Impossibility of adaptivity and finality {\cite[Thm.~4.1]{resourcepool}}]
\label{thm:cap}
No protocol is both adaptive and has finality.
\end{theorem}

The intuition is that, in the unsized and partially synchronous setting, an extended period of network asynchrony is indistinguishable (from the protocol’s local perspective) from a sharp drop in the total resource pool. Liveness in the unsized setting forces the protocol to keep confirming blocks even during such periods, while security (and hence finality) forbids confirmation during sufficiently bad partitions. These two requirements are incompatible, yielding Theorem~\ref{thm:cap}.

\paragraph{Consequences for pBFT and Nakamoto-style protocols.}
This theorem mathematically formalizes the informal dichotomy:

\begin{itemize}
  \item BFT-style protocols (including pBFT-like protocols and Algorand) operate in a sized setting and achieve finality, but they are \emph{not} adaptive: if participation drops below the designed threshold (e.g., insufficient committee signatures), the protocol halts and liveness fails.
  \item Longest-chain Nakamoto-style protocols (Bitcoin, Snow White, etc.) are adaptive in the unsized setting but, by Theorem~\ref{thm:cap}, cannot have finality. They only offer probabilistic finality in the sense of exponentially decreasing reorg probability as confirmations accumulate.
\end{itemize}

In particular, a protocol cannot simultaneously enjoy the “never reorg after confirmation, even if everyone goes offline” guarantee of pBFT-style finality and the “keeps making progress under arbitrary participation fluctuations” guarantee of Nakamoto-style adaptivity. Any design must choose a point on this fundamental trade-off.

\subsection{Ebb n Flow Protocols}

Since the properties of availability and finality cannot co-exist in a single consensus protocol, Nue, Tas and Tse \cite{ebbnflow} formulated a new class of flexible consensus protocols called \textit{ebb-n-flow protocols}.
\begin{definition}[Ebb-n-Flow Protocols]
    An \textit{ebb-n-flow protocol} is a flexible consensus protocol that supports a \textit{conservative} client and an \textit{aggressive} client. The conservative client trusts a \textit{finalized} ledger $\text{LOG}_{\text{fin}}$ which is safe under network partition and the aggressive client trusts an \textit{available} ledger $\text{LOG}_{\text{da}}$ that is live under dynamic participation.
\end{definition}
\begin{figure}[h]
    \centering
    \includegraphics[width=1\linewidth]{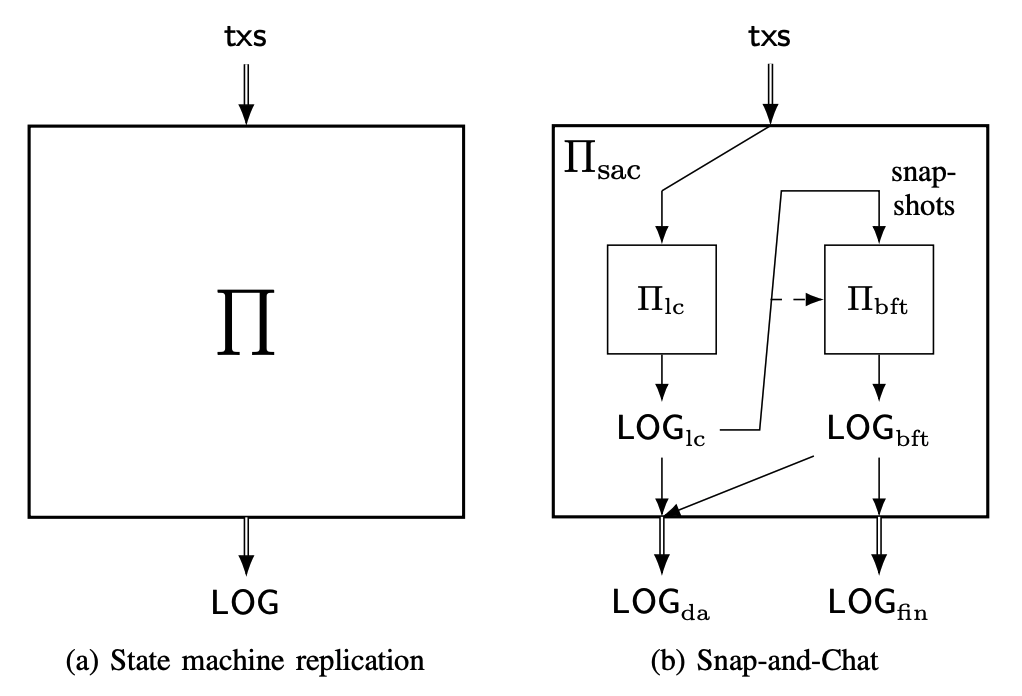}
    \caption{(a) uses a simple SMR protocol $\Pi$ to output a single ledger $LOG$ (b) uses an ebb-n-flow protocol $\Pi_{sac}$ made up of an availability protocol $\Pi_{lc}$ and a finality protocol $\Pi_{bft}$ which outputs an availability $LOG_{da}$ and finality $LOG_{fin}$ ledger respectively \cite{ebbnflow}}
    \label{fig:ebbnflow_diagram}
\end{figure}
In such a protocol, $\text{LOG}_{\text{fin}}$ is a prefix of $\text{LOG}_{\text{da}}$, i.e, $\text{LOG}_{\text{fin}}$ must always be a contiguous, initial segment of $\text{LOG}_{\text{da}}$. The finalized ledger may stagnate during network partitions, but catches up to the availability ledger when the network stabilizes. This model removes the need for a protocol-wide mandate on security or liveness, enabling clients to independently select the ledgers for desired properties.
\\\\
Ethereum's consensus protocol \textit{Gasper} is an ebb-n-flow protocol that uses LMD-GHOST as its availability protocol and Casper as its finality gadget. It considers a partially synchronous network, where honest nodes are sleepy until \textit{global awake time} $\text{GAT}$, \cite{ebbnflow} after which all nodes are awake. It also assumes that adversaries are always awake. Then
\begin{itemize}
    \item $\text{LOG}_{\text{fin}}$ is safe at all times and live after $\text{max(GST, GAT)}$, provided that less than 33\% of nodes are adversarial.
    \item If $\text{GST}=0$, $\text{LOG}_{\text{da}}$ is safe and live at all times, provided that less than 50\% of awake nodes are adversarial.
\end{itemize}

\section{The Ethereum Consensus Protocol}

Ethereum is the most decentralized blockchain in the world. It uses a Proof of Stake model as its consensus protocol. This model is called \textit{Gasper} \cite{gasper}, an ebb-n-flow protocol which can provide safety, plausible liveness, and probabilistic liveness under different assumptions. Gasper combines \textit{Casper} (the friendly-finality gadget) \cite{casper} and \textit{LMD-GHOST} (fork-choice rule).
\\\\
In the Ethereum Proof of Stake network, only nodes who stake into the network can participate in the consensus protocol. These nodes are called \textit{validators}. Validators communicate with the network by broadcasting messages with their digital signatures that can be referred as \textit{attestations}.
\\\\
Let's look at these algorithms individually and then understand how they combine together to give Ethereum the its safety and liveness guarantees.

\subsection{LMD-GHOST: The Fork-Choice Rule of Gasper}

LMD-GHOST is the fork-choice rule for the Ethereum consensus protocol. It comprises of two acronyms that stand for "Latest Message Driven" and "Greedy Heaviest-Observed Sub-Tree". LMD-GHOST is an adaptation of the GHOST protocol introduced by Sompolinsky and Zohar \cite{ghost}.

\begin{definition}[Fork-Choice Rule]
    A \textit{fork-choice rule} is an algorithm which, when given a view $V$ of all proposed blocks, deterministically outputs a block $B_V$ which is considered as the head of the canonical chain. Formally, we can say that a fork-choice rule $f$ for a given $V$ returns a chain from the genesis block $B_{genesis}$ to the leaf block $B_V$.
\end{definition}
For example, the \textit{longest-chain rule} used in the Nakamoto consensus is a fork-choice rule that returns a $B_V$ which is furthest from $B_{genesis}$.
\\\\
Each validator in the Ethereum network has its own local view, and they make attestations on blocks that they consider valid according to their view. Each attestation on a block can be considered as a \textit{vote} which has a \textit{weight}.

\begin{definition}[Weight]
    The \textit{weight} of a single vote or attestation on a block is the \textit{effective balance} of the validator which cast the vote. For a given $V$, let $M$ be a set of all \textit{latest} attestations on a block $B$ or descendants of $B$, then weight $w(V, B, M)$ is defined as the cumulative weight of each individual vote in $M$. Note that for any given $B$ a validator can only cast a single vote.
\end{definition}

LMD-GHOST is a Greedy algorithm that runs recursively on a tree. Starting at $B_{genesis}$ or any block on which the rule is to be run, we select the \textit{heaviest} branch descending from it, i.e. we choose the child block from the root which has the highest weight. In a situation where two or more child blocks have the same weight, the block with the greatest block-hash will be chosen. We then repeat the process recursively with the chosen child block as the new root until we reach a leaf. For a given $V$, this leaf block $B_V$ is the output of the algorithm and will be considered the head of the chain.
\\\\
LMD-GHOST is used as a fork-choice rule to output $\text{LOG}_\text{da}$ but it does not provide finality. Validators in the network can run LMD-GHOST to get $B_V$ but there exists a non-zero probability that $B_V$ can reorg. Gasper applies a finality gadget called Casper FFG on $\text{LOG}_\text{da}$ to output a prefix ledger $\text{LOG}_\text{fin}$ which is considered \textit{finalized}.

\subsection{Casper: The Friendly Finality Gadget}

Casper FFG is a meta-protocol that can be used on top of a consensus protocol to achieve finality. Casper operates in an asynchronous network and can tolerate at most $1/3$ of validators being faulty just like pBFT. 
\\\\
Casper FFG measures the chain in terms of \textit{height}. Each block $B$ in a tree has a height $h_B$, which is defined by its distance from the root or $B_{genesis}$. The genesis block has a height of 0. Casper FFG introduces the concepts of \textit{justification} and \textit{finalization} of blocks which is analogous to the \textit{prepare} and \textit{commit} stages in pBFT. More accurately, Casper FFG finalizes \textit{checkpoints}.

\begin{definition}[Checkpoints]
    A \textit{checkpoint} is defined as a block whose height is a multiple of some constant $H$ and $B_{genesis}$ is the first checkpoint. Two checkpoints are said to be \textit{conflicting} if and only if they are nodes in distinct branches.
\end{definition}

Now that we understand checkpoint blocks, we can formally define \textit{attestation}, a term that we have been using a bit loosely until now.

\begin{definition}[Attestation]
    An \textit{attestation} is a digitally signed message which contains a link from one checkpoint to another $s \rightarrow t$ where $s$ and $t$ are not conflicting checkpoints and $h_t > h_s$.
\end{definition}

Here, $s$ is called the \textit{source} checkpoint and $t$ is called the \textit{target} checkpoint. A link $s \rightarrow t$ is said to be a \textit{supermajority link} if at least $\frac{2}{3}$ of validators (by deposit) have made an attestation with $s$ as the source and $t$ as the target.

\begin{itemize}
    \item A checkpoint $t$ is said to be \textit{justified} if there exists a supermajority link $s \rightarrow t$ where $s$ is a \textit{justified} checkpoint.

    \item A checkpoint $t$ is said to be \textit{finalized} if $t$ is \textit{justified} and their exists a supermajority link $s \rightarrow t$ where $s$ is a \textit{finalized} checkpoint and $h_s = h_t + 1$.

    \item $B_{genesis}$ is both \textit{justified} and \textit{finalized}.
\end{itemize}

In this setting, a validator in the network is said to violate the protocol if it makes two distinct attestations corresponding to two links $s_1 \rightarrow t_1$ and $s_2 \rightarrow t_2$ such that

\begin{itemize}
    \item $h_{t_1}=h_{t_2}$
    \item $h_{s_1}<h_{s_2}<h_{t_2}<h_{t_1}$
\end{itemize}

If a validator violates any of the two conditions, its stake in the protocol is \textit{slashed}. This means that the above two conditions are never true as long as less than $\frac{1}{3}$ of the validators are adversarial.
\\\\
With the above conditions, Casper FFG proves two fundamental properties:
\begin{theorem}[Accountable Safety]
    Two conflicting checkpoints $A$ and $B$ cannot be simultaneously \textit{finalized} unless more than $\frac{1}{3}$ of total validator stake has been slashed.
\end{theorem}

\begin{theorem}[Plausible Liveness]
    New checkpoints can always be \textit{finalized} as long as new blocks are being produced and more than $\frac{2}{3}$ of total validator stake follows the protocol.
\end{theorem}

\subsection{Gasper: Ethereum's Ebb-n-Flow Protocol}

Gasper is the result of the combination of LMD-GHOST and Casper FFG. It is the Proof of Stake protocol used by Ethereum validators to achieve consensus on proposing and finalizing blocks. In Proof of Stake, the voting power (from which validator vote weight is calculated) is proportional to the amount of ETH staked by the validator into the protocol.
\\\\
Gasper does not process messages and attestations from all validators in any given stage before proceeding to the next. This is because the Ethereum network has a very large set of validators which make the traditional pBFT type finality infeasible.
\\\\
Gasper sets up time in units of \textit{slots} and \textit{epochs}, and divides the entire validator set into \textit{committees}. An epoch has 32 slots, each 12 seconds long. At the start of each epoch, the entire validator set is randomly divided into 32 committees of roughly equal size. After each epoch, the validators are shuffled between all the committees. Validators from each committee are scheduled to cast a vote in each slot which ensures that the entire validator set casts a unique vote per epoch.
\\\\
Note that this is slightly different from Casper FFG which defined checkpoints in terms of block height. Gasper on the other hand deals with epochs which can be considered as checkpoints where $H=32$.
\begin{definition}[Epoch Boundary Pairs]
    In Gasper, blocks are chosen to play the role of what \textit{checkpoints} do in Casper. One block is chosen per epoch but it is possible that multiple epochs have the same block. To distinguish the block, for an epoch $j$ and block $B$, we define an \textit{epoch boundary pair} $(B, j)$.
\end{definition}

Epoch boundary pairs, or \textit{pairs} for short, act as the checkpoints in Gasper. These pairs get \textit{justified} and \textit{finalized} just like checkpoints i.e. when a supermajority link between 2 pairs is established. The reason for this slight modification between Casper and Gasper is because Casper is a finality gadget that can be applied on a blockchain with \textit{probabilistic liveness} (in good synchrony conditions) i.e. blocks are always being produced. Gasper wanted to have less assumptions and wanted the finality gadget to work even under asynchronous conditions i.e. block production is halted.
\\\\
Note that in Gasper, when a validator makes a vote, it simultaneously makes a GHOST vote and a Casper attestation. The GHOST vote is made on the block in the slot corresponding to the committee of the validator. And it also makes a link between the last justified pair to the latest epoch boundary pair. The block of this pair must be a parent of the block the validator is making a GHOST vote for.
\\\\
Gasper also uses a slightly modified fork-choice rule HLMD GHOST (which stands for Hybrid LMD-GHOST). HLMD deviates by considering the latest justified pair $(B_J, j)$ as the root of the algorithm which outputs leaf $B_l$. In order to ensure that this view is "frozen", the algorithm does not run with the whole view $V$ but a filtered or auxiliary version $V''$. $V''$ eliminates all branches where $B_l$ is not in the set of branches where $(B_j, j)$ is a part of the justified view.
\\\\
With these slight changes, Gasper ensures that Ethereum can run a Proof of Stake protocol with Accountable Safety and Plausible Liveness. It also guarantees Probabilistic Liveness under synchronous conditions.

\subsection{Catastrophic Crash}

Gasper is supposed to provide finality to the protocol as long as more than $\frac{2}{3}$ of the network is honest and online. However, if there is a case when more than $\frac{1}{3}$ suffers a crash fault (or coordinates an attack) then block finalization halts. The network still produces blocks, but new epoch pairs will not be finalized because a supermajority link cannot be established.
\\\\
In such a catastrophic situation, Ethereum enters an \textit{inactivity leak} mode. This is triggered if an epoch boundary pair has not been finalized for some minimum predefined number of epochs. During an inactivity leak, inactive validators start to get penalized. The basic concept is that as inactive stake in the network is drained, the proportion of stake that is active increases until a supermajority link can be established, after which epoch pairs can get finalized.

\subsection{Attacks on Gasper}

Gasper requires anywhere from 64 to 95 slots to finalize blocks. This means that any transaction on Ethereum can take 12.8 to 19 minutes. On average it takes 12-15 minutes to finalize any given block. This is a long amount of time during which blocks are subject to reorganizations. There exists a possibility for validators to execute a \textit{reorg attack} to extract MEV which can disincentivize honestly following the protocol.
\\\\
Balancing attacks have been discovered \cite{attacks} that can allow a validator controlling a small fraction of total network stake and no control over network delay to achieve long-range reorg attacks.
\\\\
A reorg attack typically refers to an attempt by an adversary to fork out a block it has observed. This is predominant in PoW blockchains where the adversary wants to get the MEV in the block for themselves. This is called a \textit{ex post} attack. However, Gasper allows for a type of reorg attack called \textit{ex ante} attack. Here, an adversary can reorg future blocks that are unknown to the adversary at the time of the attack. This can be done to reorg $k$ number of blocks (also called $k-reorg$), which can give the adversary $12k$ seconds to propose a block and achieve more MEV.
\begin{figure}[h]
    \centering
    \includegraphics[width=1\linewidth]{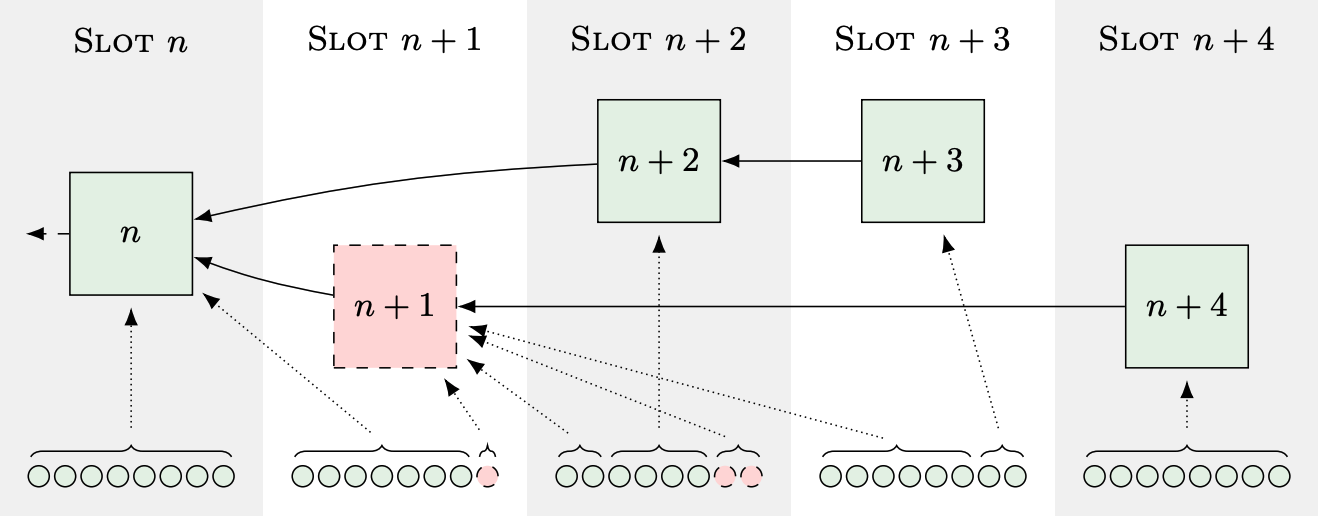}
    \caption{Example of a 2-reorg attack \cite{attacks}}
    \label{fig:attacks_diagram}
\end{figure}
A chain reorganization of this type can be executed even when only probabilistic network delays are assumed. Here is a detailed example outlining a 2-block reorganization (2-reorg):
\begin{enumerate}
    \item  In slot $n + 1$, the adversary privately proposes block $n+1$ atop block $n$ and casts a corresponding attestation for it. Both the block and the attestation are withheld to facilitate a subsequent reorganization.
    \item The proposer of slot $n + 2$, lacking visibility of the adversarial block $n + 1$, builds atop block $n$. The adversary then executes a 'balancing attack' by releasing block $n + 1$ and the withheld attestation mid-slot. This forces a split in the slot $n + 2$ attestation committee: one fraction votes for block $n + 2$ (the public head), while the other fraction observes the adversarial weight and pivots to block $n + 1$.
    \item  In slot $n+3$, the proposer builds atop block $n+2$ based on their current view of the chain. The adversary then publishes two latent attestations to the network; since these are processed by a majority of the slot $n+3$ committee prior to attesting, the LMD-GHOST weight tips in favor of the $n+1$ branch. This forces a majority vote for block $n+1$ as the current head.
    \item By slot $n + 4$, the fork-choice rule favors the $n + 1$ branch, leading the proposer to extend block $n + 1$ with block $n + 4$. This concludes the reorganization attack, where the sequence $\{n+2, n+3\}$ is supplanted by the adversarial branch $\{n+1, n+4\}$, resulting in the orphaning of the former.
\end{enumerate}

If an adversary has control of network delays, then $2k-1$ adversarial validators are enough to pull off a k-reorg attack. If the network delay is probabilistic instead, a small fraction of adversarial validators are enough to pull off an attack.

\section{Path towards Single-Slot Finality}

Since the merge, Ethereum has been producing blocks without any downtime. During the same period, its stablecoin market cap has doubled. With all its success, the underlying protocol Gasper still has several issues. Its susceptibility to long range balancing attacks (like the ex ante attack we discussed) and long finality times are points of concern.
\\\\
There is active research on consensus protocols to improve properties of Gasper or replace it with something better. Improvements on dynamically available protocols like LMD-GHOST can help us unlock better properties for the underlying protocol. The initial goal of the foundation with these efforts is to achieve finality within a single block or \textit{single-slot finality}.

\subsection{Goldfish: A drop in replacement for LMD-GHOST}

The balancing and ex-ante attacks possible on Gasper are primarily because not all attesters share the same view as the proposer at the time of the attestation. To address these issues, a new dynamically available protocol was proposed by D'Amato, Nue, Tas and Tse called \textit{Goldfish} \cite{goldfish}.
\\\\
Goldfish was proposed as a drop-in replacement for LMD-GHOST in Gasper which came with security guarantees in the sleepy model for \textit{reorg resilience}. It works on the concepts of \textit{message buffering} (also known as \textit{view-merge}) and \textit{vote expiry}.
\begin{definition}[Message Buffering]
    In LMD-GHOST, all messages received by a validator comprise of their view. With \textit{message buffering}, messages received by a validator $V$ enter a buffer $\mathcal{B}_V$. These messages can then be carefully timed for inclusion in the local view $\mathcal{T}_V$ which leads to the protocol being reorg resilient.
\end{definition}
\begin{definition}[Vote Expiry]
    \textit{Vote expiry} means that in any given slot $t$, only votes from the slot $t-1$ will be considered to run the GHOST fork-choice rule on. This prevents adversaries from accumulating weight across committees.
\end{definition}

Assuming an incurred delay of $\Delta$, Goldfish consists of synchronized slots of $3t\Delta$ time. This way, a slot is divided into 3 phases which for some slot $t$ are defined as follows:
\begin{enumerate}
    \item \textit{PROPOSE} ($3t\Delta$): In this phase, the validator that is chosen to be the proposer runs the GHOST rule using only votes from slot $t-1$. It also \textit{merges} its $\mathcal{T}$ with its $\mathcal{B}$. It then broadcasts a \textit{proposal} containing the proposed block $B_{leader}$ and the merged view $\mathcal{T}'$.
    \item \textit{VOTE} ($3t\Delta+\Delta$): In this phase, all the validators in subcommittee acknowledge the leader and merge $\mathcal{T}'$ into their own $\mathcal{T}_V$. They then run the GHOST rule using only votes from slot $t-1$ to find the tip and broadcast their vote.
    \item \textit{COMMIT} ($3t\Delta+2\Delta$): In this phase, all validators merge their $\mathcal{B}_V$ into their $\mathcal{T}_V'$ and then finally run GHOST rule only with votes from slot $t$ to output the confirmed ledger.
\end{enumerate}

These properties combined ensure that Goldfish is not prone to any balancing or long-range reorg attacks. It is also composable with finality gadgets like Casper and hence can act as a drop-in replacement for LMD-GHOST in Gasper.
\\\\
However, Goldfish cannot be considered practical because of the brittle nature of vote expiry during temporary asynchronous periods. Even a single slot of asynchrony can cause catastrophic failure as the votes arrive in the buffer after being expired. Hence, Goldfish is not \textit{asynchrony resilient} which is a crucial property for any dynamically available protocol.

\subsection{RLMD-GHOST: Providing Asynchrony Resilience}

Since it is possible for an adversary to exploit the votes of an offline validator in LMD-GHOST, it cannot be considered fully dynamically available. Goldfish's attempt to restore dynamic availability led to the protocol not being asynchrony resilient.
\\\\
To address the challenges of asynchrony resilience, D'Amato and Zanolini introduced RLMD-GHOST \cite{rlmd-ghost} which is a dynamically available protocol with asynchrony resilience. Note that challenges to asynchrony arise only because of the strict vote expiry proposed in Goldfish. RLMD-GHOST relaxes the vote expiry to balance dynamic availability with asynchrony resilience. 
\\\\
It is important to note that balancing attacks with accumulation of weight are only possible due to \textit{subsampling} of the validator set in LMD-GHOST. In order to prevent the re-introduction of ex-ante attacks in RLMD-GHOST, it does away with subsampling, opting for a single validator set that votes in a slot.
\\\\
Recent Latest Message Driven GHOST is introduced as a family of \textit{propose-vote-merge} protocols that are parameterized by a vote expiry period $\eta$. We can see that when we set $\eta=1$, the protocol reduces to a variant of Goldfish without subsampling. And when we set $\eta=\infty$, the protocol reduces to a variant of LMD-GHOST without subsampling and with view-merge.
\\\\
D'Amato and Zanolini also specify the joining protocol for the propose-vote-merge family \cite{rlmd-ghost}. When a sleepy validator $V$ wakes up in slot $t$, it receives all the messages while it was asleep and adds them to its buffer $\mathcal{B}_V$. It then waits for the next merge phase $3t\Delta+2\Delta$ to merge its $\mathcal{B}_V$ into its view $\mathcal{T}_V$. Until this point, the validator is considered \textit{dreamy} \cite{goldfish} and only after the merge is it allowed to participate in the protocol as a regular \textit{awake} node.
\\\\
RLMD-GHOST is a family of protocols that allow us to balance between asynchrony resilience and dynamic availability by tweaking the parameter $\eta$. As $\eta$ grows, the protocol becomes resilient to longer periods of asynchrony but it also allows the inclusion of votes of offline validators up to $\eta$ slots old, affecting the degree of dynamic availability.

\subsection{Generalized Sleepy Model}

RLMD\text{-}GHOST is analyzed in the \emph{generalized} sleepy model introduced by D'Amato and Zanolini~\cite{rlmd-ghost}, which refines the original Pass--Shi sleepy model~\cite{sleepy} by accounting for the age of votes that may be safely incorporated into a validator’s view. This refinement allows us to quantify precisely how the vote-expiry parameter $\eta$ influences both \emph{dynamic availability} and \emph{asynchrony resilience}.
\\\\
Time is divided into slots. Each validator $V$ may be in one of three modes: awake, offline and dreamy. The transitions are defined as follows:

\begin{itemize}
    \item \textbf{Awake:} Validator $V$ participates fully in the protocol. It produces votes in slot $t$, receives all messages sent during slot $t$ with delay at most $\Delta$, and merges its buffer into its view at designated merge points.
    \item \textbf{Offline:} Validator $V$ produces no messages in slot $t$. From the network perspective it is indistinguishable from an adversary.
    \item \textbf{Dreamy:} Upon waking up at the beginning of a slot $t$, validator $V$ receives all messages it missed while offline, placing them into its buffer $\mathcal{B}_V$, but it cannot yet vote. It transitions to \emph{awake} only after the merge phase of slot $t$.
\end{itemize}

This induces the following fundamental constraint: an honest vote cast in slot $s$ becomes \emph{usable} by an honest validator that wakes up in slot \(t\) only if 
\[
t - s \le \eta,
\]
where $\eta$ is the protocol’s vote-expiry bound.

The original sleepy model~\cite{sleepy} requires only that the majority of \emph{online} validators remain honest in each slot. The generalized model quantifies how much ``stale'' information can be safely incorporated into the fork-choice rule via $\eta$.

\begin{definition}[Degree of Dynamic Availability]
\label{def:degree-da}
A protocol has \emph{degree-$\eta$ dynamic availability} if, for every slot $t$, liveness is preserved whenever the set of honest validators that were awake at least once in the window
\[
[t-\eta,\, t-1]
\]
have a majority over adversarial validators that were awake in the same window.
\end{definition}

In other words, the protocol can tolerate validators being offline for up to $\eta$ slots, while still treating their votes as ``fresh'' and counting them toward honest weight during fork-choice execution. Increasing $\eta$ strictly increases the set of adversarial-free executions in which liveness is preserved but simultaneously expands the adversary’s ability to exploit stale votes.

This trade-off is the basis for the RLMD-GHOST design space: setting
\[
\eta = \pi + 2
\]
gives resilience against asynchronous periods of length $\pi$ while preserving dynamic availability for validators that go offline for at most $\eta$ slots, but choosing $\eta$ too large reintroduces long-range balancing attacks known from subsampled LMD-GHOST.

\subsection{Single-Slot Finality on Ethereum}

Gasper still takes between 64 and 95 slots to finalize blocks. This is a huge concern for several applications which require faster confirmations. Building on top of RLMD-GHOST, D'Amato and Zanolini proposed a simple single-slot finality protocol \cite{ssf} which combines the synchronous dynamically available protocol with a partially synchronous finality gadget.
\\\\
The protocol uses a hybrid version of RLMD-GHOST as its fork choice rule, called $HFC$. It is similar to RLMD-GHOST but it starts from the latest justified block $LJ(\mathcal{T})$ in $\mathcal{T}$ instead of $B_{genesis}$. The finality gadget is similar to Casper FFG, where a supermajority link between two checkpoints justifies and finalizes the target and source respectively.
\\\\
The proposed protocol is divided into 4 phases:
\begin{enumerate}
    \item \textit{PROPOSE} ($4t\Delta$): In this phase, the validator that is chosen to be the proposer runs the $HFC$ rule using votes from slots $[t-\eta,t-1]$. It also \textit{merges} its $\mathcal{T}$ with its $\mathcal{B}$. It then broadcasts a \textit{proposal} containing the proposed block $B_{leader}$ and the merged view $\mathcal{T}'$. 
    \item \textit{HEAD-VOTE} ($4t\Delta+\Delta$): In this phase, the validators merge $\mathcal{T}'$ into their own $\mathcal{T}_V$. They then run the HFC rule using votes from slots $[t-\eta,t-1]$ to find the tip and broadcast their HEAD-VOTE (similar to GHOST vote).
    \item \textit{FFG-VOTE} ($4t\Delta+2\Delta$): In this phase, the validators runs HFC rule to find tip $B_{t}$ with a supermajority of votes from slot $t$ for fast confirmation. Then the validator casts an FFG-VOTE linking $LJ(\mathcal{T'})$ to $B_{t}$ and broadcasts it.
    \item \textit{MERGE} ($4t\Delta+3\Delta$): In this phase, all validators merge their $\mathcal{B}_V$ into their $\mathcal{T}_V'$.
\end{enumerate}
\begin{figure} [h]
    \centering
    \includegraphics[width=1\linewidth]{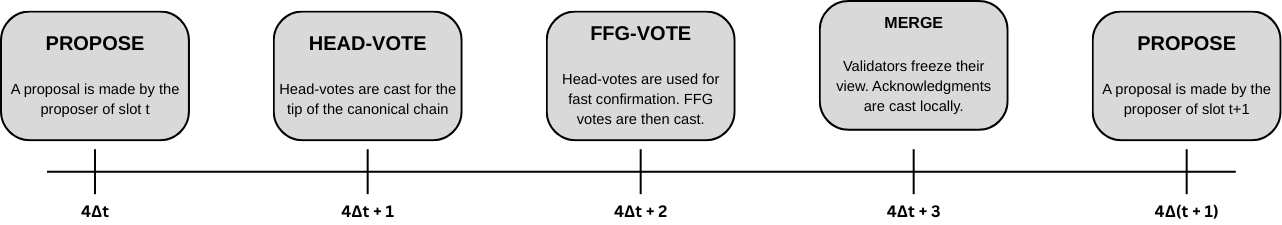}
    \caption{Four phases of slot $t$ with SSF}
    \label{fig:ssf}
\end{figure}
This way, we achieve a secure and reorg resilient ebb-n-flow that can finalize block in a single slot. Theoretically it can replace Gasper as Ethereum's consensus protocol and allow single slot finality and strong security guarantees to the network. However, due to the large number of validators in the Ethereum network, the time it takes to aggregate all validator signatures makes the implementation of a single slot finality protocol impractical.

\section{Endgame for Fast Finality}

The simple single slot finality protocol aims to replace Gasper with a Tendermint like protocol which also has availability guarantees due to its ebb-n-flow nature and inclusion of "inactivity leak". The most challenging aspect of this proposal is that each validator broadcasts two messages per slot. 
\\\\
As Vitalik points out in his blog \cite{vitalik}, \\
\textit{``There is a deep philosophical reason why epoch-and-slot architectures seem to be so hard to avoid: it inherently takes less time to come to approximate agreement on something, than to come to maximally-hardened "economic finality" agreement on it."}
\\\\
As the number of nodes in the network increases, which is a key metric of decentralization and an important factor for Ethereum, the amount of time required for gathering signatures (aggregating and distributing) becomes the limiting factor for slot times. This means that a voting phase takes longer (double the normal network latency) than a non-voting phase in a slot. Hence, it is beneficial to explore all forms of slot and epoch based architectures. This is where Vitalik also retcons the acronym SSF to "Speedy Secure Finality" to allow a larger design space for faster finality.
\\\\
Reducing slot times also falls under this goal. One of the most well understood protocols that aim at providing fast finality while also reducing slot times is \textit{3-Slot Finality protocol} or \textit{3SF} proposed by D’Amato, Saltini, Tran and Zanolini. \cite{3sf} The original SSF protocol had two voting phases per slot while 3SF proposes just one voting phase per slot while maintaining all safety and liveness guarantees of SSF.
\\\\
To clarify the structural commonality between recent fast-finality designs, we present a
high-level propose-vote-merge abstraction. This template is not intended as a protocol
specification, but as a normal form capturing the minimal control parameters that differentiate
existing constructions.

\begin{figure}[t]
\centering
\fbox{
\begin{minipage}{0.95\linewidth}
\textbf{Propose--Vote--Merge Template for Speedy Secure Finality}

\vspace{0.5em}
\textbf{State:} local view $\mathcal{T}$, message buffer $\mathcal{B}$ \\
\textbf{Parameters:} vote window $\eta \in \mathbb{N}\cup\{\infty\}$, voting rounds $r \ge 1$

\vspace{0.75em}
\textbf{PROPOSE:}
\begin{itemize}
  \item Merge buffered messages: $\mathcal{T} \leftarrow \mathsf{Merge}(\mathcal{T},\mathcal{B})$
  \item Compute head: $B \leftarrow \mathsf{ForkChoice}(\mathcal{T},\eta)$
  \item Broadcast $\mathsf{Propose}(B,\mathcal{T})$
\end{itemize}

\textbf{VOTE ($r$ rounds):}
\begin{itemize}
  \item Collect votes from slots $[t-\eta,\,t-1]$
  \item Update head / justification state
  \item Broadcast vote message
\end{itemize}

\textbf{MERGE:}
\begin{itemize}
  \item Finalize view $\mathcal{T}$
  \item Discard expired messages from $\mathcal{B}$
\end{itemize}
\end{minipage}
}
\caption{A unifying propose--vote--merge abstraction underlying modern fast-finality protocols.}
\label{fig:pvm}
\end{figure}

Within this abstraction, Gasper corresponds to an epoch-subsampled instantiation with effectively
unbounded vote expiry. Single-slot finality corresponds to $r=2$ with bounded $\eta$, while
3-slot finality collapses the voting phases to $r=1$. This perspective makes explicit that 3SF
reduces per-slot communication without introducing new consensus primitives.

% \begin{algorithm}[t]
% \caption{Propose--Vote--Merge Template for Speedy Secure Finality}
% \label{alg:pvm}
% \begin{algorithmic}[1]
% \State \textbf{State:} local view $\mathcal{T}$, message buffer $\mathcal{B}$
% \State \textbf{Parameters:} vote window $\eta \in \mathbb{N}\cup\{\infty\}$, number of vote rounds $r \ge 1$
% \Statex

% \State \textbf{PROPOSE:}
% \State $\mathcal{T} \leftarrow \mathsf{Merge}(\mathcal{T}, \mathcal{B})$
% \State $B \leftarrow \mathsf{ForkChoice}(\mathcal{T}, \eta)$
% \State broadcast $\mathsf{Propose}(B, \mathcal{T})$
% \Statex

% \For{$i = 1$ to $r$} \Comment{Voting rounds}
%     \State \textbf{VOTE$_i$:}
%     \State collect votes from window $[t-\eta,\,t-1]$
%     \State update head / justification state using votes
%     \State broadcast vote message
% \EndFor
% \Statex

% \State \textbf{MERGE:}
% \State $\mathcal{T} \leftarrow \mathsf{Finalize}(\mathcal{T})$
% \State clear expired messages from $\mathcal{B}$
% \end{algorithmic}
% \end{algorithm}

\subsection{3-Slot Finality (3SF)}

The 3SF protocol represents a deliberate retreat from the objective of single-slot finality in favor of a design that is compatible with the dominant engineering bottleneck of large-scale proof-of-stake systems: signature aggregation latency. Rather than optimizing for the minimum theoretical finalization depth, 3SF optimizes for slot length under realistic communication assumptions, while preserving the accountable safety and liveness guarantees of SSF.
\\\\
At a high level, 3SF merges the \textit{head vote} and \textit{FFG vote} phases of SSF into a \emph{single voting phase per slot}. This reduction removes one full signature-aggregation round from the critical path of each slot. The resulting protocol finalizes blocks deterministically within three consecutive slots under synchrony and sufficient honest participation.
\\\\
In SSF, a slot contains two voting phases, each of which under aggregation requires approximately $2\Delta$ time. Consequently, the effective slot length is dominated by voting rather than proposal or merge phases. 3SF observes that finality latency and slot latency are distinct optimization targets: sacrificing same-slot finality allows one to shorten the slot itself, which in turn reduces MEV extraction windows and improves throughput-sensitive applications such as AMMs.
\\\\
Time is divided into slots indexed by $t$, each of duration $4\Delta$, subdivided into four phases:

\begin{enumerate}
    \item $\mathsf{PROPOSE} ([4t\Delta,\,4t\Delta+\Delta)$).  
    The slot proposer computes the available chain head using a dynamically available fork-choice rule like RLMD-GHOST\footnote{The original 3SF paper described two protocols, one which used RLMD-GHOST as it's fork-choice rule and another using a quorum based dynamically available protocol called TOB-SVD\cite{tobsvd}. The quorum based variant was later spun out from 3SF with some algorithmic tweaks as an independent protocol called Majorum\cite{majorum}.} applied to its frozen view, and broadcasts a proposal extending that chain. As in SSF, vote expiry applies to the \emph{combined} vote described below.

    \item $\mathsf{VOTE} ([4t\Delta+\Delta,\,4t\Delta+2\Delta)$).  
    Each active validator emits a \emph{single} vote message that simultaneously serves two roles:
    \begin{itemize}
        \item a head vote for the dynamically available chain, and
        \item an FFG-style vote linking the greatest justified checkpoint in its view to the voted head.
    \end{itemize}
    Formally, a validator broadcasts
    \[
        \textsf{vote}(ch,\, C_s \rightarrow C_t),
    \]
    where $C_s$ is the greatest justified checkpoint and $C_t=(ch,t)$ is a target checkpoint for the current slot.

    \item $\mathsf{CONFIRM} ([4t\Delta+2\Delta,\,4t\Delta+3\Delta)$).  
    Validators check whether a supermajority of votes supports a consistent chain prefix. If so, the locally confirmed chain is updated, provided it is a prefix of the available chain.

    \item $\mathsf{MERGE} ([4t\Delta+3\Delta,\,4(t+1)\Delta)$).  
    Buffered messages are merged into the local view, finalizing the state for the next slot.
\end{enumerate}

Let $B_t$ denote the block proposed in slot $t$.
\begin{itemize}
    \item In slot $t+1$, if a supermajority link $B_t \rightarrow B_{t+1}$ is observed, $B_t$ becomes \emph{justified}.
    \item In slot $t+2$, if $B_{t+1}$ is justified and a supermajority link $B_{t+1} \rightarrow B_{t+2}$ is formed, then $B_t$ is finalized.
\end{itemize}
Thus, blocks proposed by honest validators are finalized within three slots, independent of proposer honesty in intermediate slots.
\begin{figure}
    \centering
    \includegraphics[width=1\linewidth]{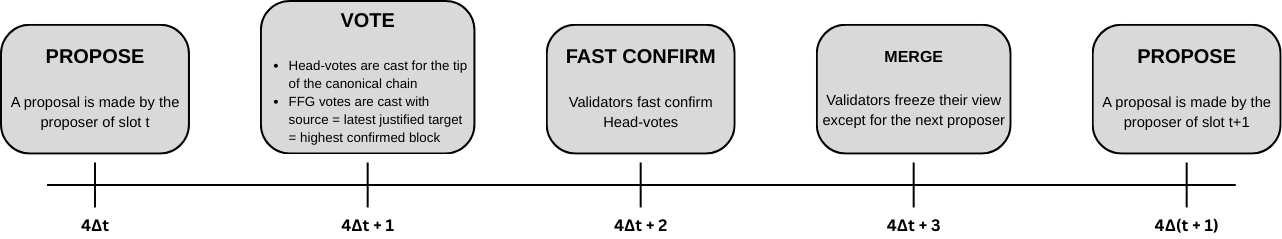}
    \caption{Four phases of slot $t$ with 3SF}
    \label{fig:3sf}
\end{figure}
\\\\
A subtle but fundamental difference between 3SF and Gasper or SSF arises from the merged vote structure. In Casper FFG and SSF, it is impossible for two distinct checkpoints from the same slot to be justified. In 3SF, this invariant no longer holds: validators may justify multiple checkpoints per slot unless constrained.

To restore accountable safety, 3SF introduces an additional monotonicity condition.

\begin{definition}[3SF Monotonicity Slashing Condition]
A validator is slashable if it produces two votes corresponding to links
\[
    s_1 \rightarrow t_1 \quad \text{and} \quad s_2 \rightarrow t_2
\]
such that
\begin{enumerate}
    \item $\mathrm{slot}(s_1) < \mathrm{slot}(s_2)$, and
    \item $\mathrm{height}(t_2) < \mathrm{height}(t_1)$.
\end{enumerate}
\end{definition}

Intuitively, this condition enforces that a validator’s justified targets must advance monotonically with time. This rule has no analogue in Casper because Casper’s voting structure makes such regressions impossible by construction.
\\\\
With this additional constraint, 3SF achieves the same accountable safety and plausible liveness guarantees as SSF and Gasper under the standard $f < \tfrac{n}{3}$ fault assumption. Conflicting finalized chains imply slashing of at least one-third of total stake, while liveness holds during sufficiently long synchronous periods with adequate honest participation.
\\\\
A key conceptual insight exposed by 3SF is that \emph{finality depth is not the sole driver of economic security}. Reducing slot duration can shrink MEV extraction windows even when finality itself is delayed by additional slots. In this sense, 3SF occupies a distinct point in the design space: it trades confirmation depth for temporal density of blocks.
\\\\
Moreover, 3SF demonstrates that finality gadgets need not be optimized solely for worst-case latency. Instead, they can be co-designed with aggregation and networking realities, yielding protocols that are asymptotically less elegant than SSF but strictly superior under real-world constraints.
\\\\
Viewed holistically, SSF, two-slot variants, and 3SF can be compared as follows:
\begin{itemize}
    \item SSF minimizes finality depth but maximizes per-slot communication.
    \item 3SF minimizes per-slot communication but increases finality depth.
    \item Intermediate designs interpolate between these extremes.
\end{itemize}

\subsection{Signature Aggregation as a Fundamental Bottleneck}
\label{subsec:sig-bottleneck}

We formalize the engineering limitation imposed by large validator sets and its implications for fast-finality protocols. We present a partial and intentionally abstract formalization whose role is heuristic rather than complete. It is used to guide intuition and structure the discussion of security and performance properties, without claiming coverage of all protocol variants or edge cases, particularly those arising in future 3SF-derived designs.

\begin{definition}[Validator Keys and Independent Nodes]
Let $k$ denote the number of active validator keys participating in consensus and $n$ denote the number of independently operated consensus nodes (i.e., independent failure domains)\footnote{Not to be confused with our initial definition of $n$ as total number of consensus nodes}. In Ethereum, $k \gg n$ due to stake-splitting across multiple validator keys operated by the same entity.
\end{definition}

\begin{definition}[Maximum Effective Balance]
The \emph{maximum effective balance} $B_{\max}$ is an upper bound on the voting weight of a single validator key. Any stake in excess of $B_{\max}$ must be represented via additional validator keys.
\end{definition}

At the Merge, Ethereum set $B_{\max} = 32$ ETH, implying that an entity staking $S$ ETH controls at least $\lceil S / 32 \rceil$ validator keys, regardless of whether these keys are operated on a single physical node.

\begin{theorem}[Signature Aggregation Lower Bound]
\label{thm:aggregation-bound}
In any slot-based proof-of-stake protocol where all validator keys vote in every slot, the time duration of a slot is lower-bounded by $\Omega(k)$ cryptographic signature aggregation.
\end{theorem}

\begin{proof}
Each validator key produces a distinct signature that must be (i) received, (ii) verified or aggregated, and (iii) redistributed to other participants. Even under optimistic assumptions i.e perfect aggregation trees, no adversarial delays, and constant-size aggregated signatures, the aggregation process must incorporate information from each of the $k$ keys. Therefore, the total aggregation work per slot grows at least linearly in $k$, establishing a lower bound of $\Omega(k)$.
\end{proof}
Theorem \ref{thm:fast-finality-infeasible} provides an empirical proof that fast finality protocols today are infeasible for large number of validators in Ethereum. Assume each validator has bounded inbound bandwidth $W$ (bits/sec) and bounded per slot CPU budget; i.e., there exists a constant $W$ such that no honest node can receive more than $W\tau$ bits within a slot.

\begin{theorem}[Infeasibility of Global-Vote Fast Finality under Bounded Bandwidth]
\label{thm:fast-finality-infeasible}
Let $k$ be the number of active validator keys and let $\tau>0$ be the fixed slot duration. Consider any slot-based proof-of-stake protocol in which a block can be confirmed or finalized in a slot only if an aggregate attestation incorporating votes from at least $qk$ distinct validator keys is produced in that slot, for some constant $q>0$.  

Assume that each honest participant has bounded inbound bandwidth $W$ (bits per second). If
\[
q k b \;>\; W \tau,
\]
where $b>0$ is a lower bound on the number of bits required per vote contribution, then such a protocol is infeasible. In particular, for $k \approx 10^{6}$ and $\tau=12$ seconds, global-vote fast-finality protocols without subsampling (including SSF- and 3SF-style constructions) are infeasible under current slot-time constraints.
\end{theorem}

\begin{proof}
For a block to be confirmed in a given slot, some honest participant must be able to determine that at least $qk$ distinct validator keys contributed valid votes in that slot. This requires receiving information sufficient to account for each of these $qk$ contributions.  

Let $b>0$ be a lower bound on the number of bits needed to convey one such contribution (e.g., a signature share plus identity or anti-equivocation metadata). Then any participant that verifies the aggregate must receive at least $qkb$ bits during the slot.  

Under the bounded-bandwidth assumption, no participant can receive more than $W\tau$ bits in a slot. If $qkb > W\tau$, this necessary condition cannot be met, and confirmation within the slot is impossible. The claim follows by direct contradiction.
\end{proof}

\begin{theorem}[Key Inflation without Decentralization Gain]
\label{thm:key-inflation}
Under a fixed maximum effective balance $B_{\max}$, increasing total staked ETH inflates the number of validator keys $k$ without increasing the number of independent consensus nodes $n$.
\end{theorem}

\begin{proof}
Each independent operator with stake $S$ must instantiate at least $\lceil S / B_{\max} \rceil$ validator keys. Increasing $S$ increases $k$ proportionally while leaving the operator count and thus the number of independent failure domains, unchanged. Therefore, $k$ grows independently of $n$.
\end{proof}

One of the ways to enable fast finality in this setup would be to increase the effective balance $B_{max}$ for validators. Notice that doing so will allow the validator keys to fall drastically and be closer to the number of real nodes in network today. Although doing so would be a practical challenge because of engineering and maintenance overhead faced by staker networks who manage these validator keys. Achieving this would also make the implementation of a fast finality protocol like 3SF feasible in the Ethereum network.

\begin{theorem}[Enabling Fast Finality via Effective Balance Scaling]
\label{thm:balance-scaling}
Increasing the maximum effective balance $B_{\max}$ reduces the per slot signature aggregation burden without weakening fault tolerance.
\end{theorem}

\begin{proof}
Consider an operator with total stake $S>0$. Under a maximum effective balance $B_{\max}$, the operator must instantiate at least
\[
k(S,B_{\max}) := \left\lceil \frac{S}{B_{\max}} \right\rceil
\]
validator keys. The function $k(S,B_{\max})$ is weakly decreasing in $B_{\max}$ for fixed $S$. Summing over all operators, increasing $B_{\max}$ induces a weakly decreasing total number of validator keys $k$ in the system.

The value of $n$ depends only on the partition of total stake across operators and is invariant under changes to $B_{\max}$. Hence, scaling $B_{\max}$ reduces $k$ while leaving $n$ unchanged.

Let $\alpha \in [0,1]$ denote the fraction of total stake controlled by adversarial operators. Safety and liveness guarantees of proof-of-stake consensus depend on $\alpha$ relative to fixed thresholds (e.g., $\alpha<1/3$), not on the number of validator keys used to represent that stake. Since increasing $B_{\max}$ leaves the stake distribution, and thus $\alpha$,unchanged, the fault-tolerance threshold is preserved.

By Theorem~\ref{thm:fast-finality-infeasible}, the per slot signature aggregation requirement scales at least linearly in $k$. Therefore, reducing $k$ via an increase in $B_{\max}$ strictly relaxes the aggregation constraint, enabling the feasibility of fast-finality protocols such as 3SF once $k$ is reduced to the same order of magnitude as $n$.

\end{proof}

EIP-7251\cite{maxB} raises $B_{\max}$ from 32 ETH to 2048 ETH, permitting a reduction of $K$ by up to a factor of 64 in the steady state. This shift moves Ethereum from a committee-mandated architecture toward one where global-vote fast-finality protocols are practically viable.

\section{Conclusion}

The pursuit of Speedy Secure Finality represents a pivotal shift in the evolution of permissionless consensus protocols. As we have detailed, the current state of Ethereum's consensus, Gasper, successfully creates a hybrid between dynamically available and finalized ledgers. However, its significant time-to-finality leaves the network exposed to ex-ante reorg attacks and limits the efficiency of economic settlement.
\\\\
Our analysis of recent advancements highlights a clear trajectory in solving these limitations. Goldfish demonstrated that reorg resilience is achievable through message buffering and vote expiry, though at the cost of brittleness to asynchrony. RLMD-GHOST refined this by introducing generalized sleepiness, offering a tunable trade-off between dynamic availability and asynchrony resilience. Building upon these primitives, Single-Slot Finality (SSF) protocols theoretically solve the latency problem entirely but introduce prohibitive communication overheads.
\\\\
We conclude that the 3-Slot Finality (3SF) protocol currently represents the most viable endgame for Ethereum. By pipelining the justification and finalization steps into a single voting phase per slot, 3SF strikes a pragmatic balance: it reduces finality and slot times without imposing the extreme bandwidth requirements of SSF. However, the feasibility of 3SF also relies on solving the signature aggregation bottleneck with an end goal of reducing the number of validator keys to match the number of unique independent nodes in the network as closely as possible.

\section{Future Work}

Current proposals like SSF and 3SF integrate finality and availability, tightly coupling economic guarantees with block production. The fundamental architecture of coupled slots and finality creates a bottleneck: effectively aggregating signatures from the full validator set requires either longer slot times or massive validator consolidation.
\\\\
A key direction for future research is the decoupling of these processes into asynchronous layers. In this architecture, an availability layer prioritizes liveness, using small committees to produce blocks rapidly, independent of the total network size. Simultaneously, a separate finality gadget aggregates votes from the full validator set on an independent timeline. 
\\\\
Vitalik points out such an architecture in a forum post\cite{decouple}, suggesting for committees of around 256 validators per slot rapidly fast confirm blocks with 6-8 second slot times, followed by a gadget finalizing blocks in under a minute with full validator set participation (without the need for extensive validator consolidation). The key concept of decoupling consensus is to treat $GHOST$ votes separate from $FFG$ votes.
\\\\
This separation mitigates the trade-off between validator set size and latency, allowing for fast block times despite high validator counts. Future work must formalize the security bounds of the availability layer, ensuring resistance to short-range reorgs before the slower finality mechanism confirms the chain state. 
\bibliographystyle{plain}
\bibliography{references}

\end{document}